%% file: main.tex
\documentclass[submission,copyright,creativecommons]{eptcs}

\usepackage{amsmath,amssymb}  
\usepackage{amsthm}
\usepackage{graphicx,tikz} 
\usepackage{lipsum}
\usepackage{hyperref}
\usepackage{url}

\newtheorem{definition}{Definition}[section]
\newtheorem*{definition*}{Definition}

\newtheorem{proposition}{Proposition}[section]

\newtheorem{example}{Example}[section]

\usepackage{multicol}
\usepackage[numbers]{natbib}
\usepackage{lineno}

\usepackage{pifont}

\def \mls {\textsf{MLMS}}
\def \SubLat {\mathsf{DsubLat}}
\def \LAT {\mathsf{LAT}}
\def \inc {inc}

\def \cert {cert}
\def \rinc {rinc}

\usepackage{amsmath,amssymb}  
\usepackage{graphicx,tikz} 
\usepackage{xcolor} 
\usepackage[all]{xy} 

\usepackage{cleveref}
\usepackage{lipsum}

\usepackage{float}
\usepackage{proof}

\newcommand{\vv}{\vspace{-1mm}}

\newcommand{\Lg}{\ensuremath{\mathcal{L} }}

\newcommand{\pair}[1]{\langle #1 \rangle}

 \newcommand{\lw}{\ensuremath{L4^w}}
 \newcommand{\ls}{\ensuremath{L4^s}}
 \newcommand{\nw}{\ensuremath{N3^w}}
 \newcommand{\ns}{\ensuremath{N3^s}}
 \newcommand{\bw}{\ensuremath{B3^w}}
 \newcommand{\bs}{\ensuremath{B3^s}}
 \newcommand{\cw}{\ensuremath{C2^w}}
 \newcommand{\cs}{\ensuremath{C2^s}}

\newcommand{\fdetobot}{\ensuremath{FDE^\to_{\bot}}}

\newcommand{\ktto}{\ensuremath{{\it K3}^{\to}}}

\newcommand{\lpto}{\ensuremath{{\it LP}^{\to}_\bot}}

\newcommand{\jt}{\ensuremath{\textit{J3}}}

\newcommand{\lt}{\ensuremath{\textit{\L3}}}

\newcommand{\lj}{\ensuremath{\textit{{\L}J4}}}

\label{xLOGICS}{}

\newcommand{\letkp}{\ensuremath{LET_{K}^+}}

\newcommand{\fde}{\ensuremath{FDE}}

\newcommand{\bo}{\textsf{b}}
\newcommand{\nei}{\textsf{n}}

\newcommand{\lets}{\textit{LET}s}

 \newcommand{\pp}{\ensuremath{[P]}}
\newcommand{\ff}{\ensuremath{[F]}}
\newcommand{\xx}{\ensuremath{[X]}}

\newcommand{\T}{\ensuremath{\mathcal{T}}}

\label{xCOLORS} %

\newcommand{\white}{\color{white}}

\label{xSYMBOLS}

\newcommand{\con}{\ensuremath{{\circ}}}

\label{xGENERAL COMMANDS}%

\newcommand{\mh}{\noindent}

\newcommand{\m}{\vspace{1mm}}
\newcommand{\mm}{\vspace{2mm}}
\newcommand{\mmm}{\vspace{3mm}}
\newcommand{\mmmm}{\vspace{4mm}}
\newcommand{\setl}{\setlength\itemsep{-0.2em}}

\newcommand{\bqu}{\begin{quote}} 
\newcommand{\equ}{\end{quote}}
\newcommand{\enr}{\begin{enumerate}[label={(\arabic*)}, resume]}
\newcommand{\eenr}{\end{enumerate}}

\label{xPORTUGUES}

\usepackage{multicol}
\usepackage[numbers]{natbib}
\usepackage{pifont}

\def \inc {inc}
\def \cert {cert}
\def \rinc {rinc}

 \title{Merging and Evolution of Data Bases: A Many-Logic   Approach} 
\author{Manuel A. Martins
\institute{CIDMA \& Dep. of Mathematics\\ University of Aveiro}\email{martins@ua.pt}
		        \and
 Abilio Rodrigues
\institute{Federal University of Minas Gerais\\ Minas Gerais, Brazil}
\email{abilio.rodrigues@gmail.com}
}

\def\titlerunning{Merging and Evolution of Data Bases}

\def\authorrunning{M. Martins \& A. Rodrigues}

\begin{document}

\maketitle

\begin{abstract}
A new approach to information-based logics using many-logic modal structures (\mls s) was recently proposed by Martins et al. In such a setting, it is possible to have worlds with different underlying logics within the same frame, all anchored to a base lattice that contains the semantics of each logic as a down-complete sublattice.
In this paper, we take the lattice \textit{L6}, defined by the six-valued semantics of the logic of evidence and truth \letkp, as our base lattice. We also consider four-, three-, and two-valued formal systems associated with down-complete sublattices of \textit{L6}. These logics are capable of representing paracomplete, paraconsistent, and classical contexts with six, four, three, and two informational scenarios, respectively.
    As applications of \mls s, we will discuss the merging and temporal evolution of databases whose underlying logics are two-, three-, four-, and six-valued. We will also investigate  inconsistency measures   in this context.

 \end{abstract}

\section{Introduction}

An information-based logic is a logic  suitable for processing
information  in the sense of taking a database as a set of premises and drawing
conclusions from these premises in a sensible way. 
Nuel Belnap in \cite{belnap.1977.how} proposed an information-based interpretation of the logic of first-degree entailment (\textit{FDE}), motivated by the idea of computer dealing with possibly inconsistent and incomplete databases, that is, situations in which there may be conflicting information or no information at all about a sentence $A$. In the former scenario, explosion fails; in the latter, excluded middle fails. Such situations are not exceptional, but rather a common reality, often resulting from noise, uncertainty, and inaccuracies in the process of collecting information.


The computer conceived by Belnap `modifies its current set-up' when it gets new information 
\cite[p.~36]{belnap.1977.how}. 
A `set-up', thus, is a specific configuration of the computer over time, 
and these configurations can be seen as \textit{information states} 
represented by worlds in a Kripke model.
The many-logic modal structures (\mls) proposed in  \cite{RM24} and \cite{Martins2026} are well-suited for representing the connections between information states 
viewed as  databases with different underlying logics, the evolution of databases over time, and 
databases accessed by users with different privileges.  


This framework considers frames whose worlds may be associated with different formal systems. Each of these systems is represented by a down-complete sublattice (dc-sublattice, for short) of a common base lattice. Modal relationships are then established through this shared algebraic structure, which embeds the semantics of the individual logics and provides a unified setting in which they can interact. As a consequence, connections between worlds can be understood as interactions between their underlying logics. In such a setting, the modalities are interpreted in a way that differs from the traditional Kripkean approach.

In addition to introducing the general formal machinery of \mls s\footnote{The development of \mls s has recently emerged as a framework for extending modal reasoning beyond the classical setting. The foundational proposal was introduced by Freire and Martins in \cite{RM24}, who developed the notion of modality across different logical systems. The algebraic basis of the framework was later investigated by Freire, Madeira, and Martins in \cite{Freire_Madeira_Martins}, who introduced lattice-based constructions relating the lattices associated with distinct logics. Further developments include the study of regular non-normal modal systems (\cite{freiremartins-coniglio}) and the institutional foundations of many-logic modal logic (\cite{FMM24}).
}, we will investigate a number of applications based on the lattice  {$L6$}, which provides the six-valued semantics for the logic of evidence and truth \letkp, introduced by Coniglio and Rodrigues in \cite{con.rod.2023}. 
Logics of evidence and truth (\textit{LETs}) are equipped with a classicality operator \con\ that 
recovers classical logic for sentences in its scope. Thus, when $\con A$ holds, 
the sentence $A$ is subjected to classical logic, 
otherwise the underlying logic of $A$ is \fde, or an extension of \fde.   
\lets\ can   be interpreted in terms of information, and in this case a 
sentence $\con A$ means that the information conveyed by $A$, positive or negative, is reliable. 
The connective \con, thus,  can be thought of as a kind of \textit{certification}.   
$L6$ extends the well-known logical lattice $L4$,
defined by the four-valued semantics of \textit{FDE} \cite[p.~42]{belnap.1977.how}, 
by adding a new top and a new bottom, 
which correspond to two additional semantic values intuitively interpreted as reliable positive information 
and reliable negative information (cf.~\cite[Sect.~4.1]{con.rod.2023}). \letkp, as \lets\ in general, is able to represent six scenarios: the four scenarios of \fde\ plus two more scenarios of reliable information. 
A detailed study of logics defined by dc-sublattices of \textit{L6} can be found in \cite{Martins2026}.

The aim of this article is to explore how dynamic interactions between distinct logical contexts can be formalized in richer computational settings, as a step toward extending the original proposal.
In particular, we explore mechanisms for information merging and temporalization, enabling the representation not only of static accessibility relations between worlds, but also of the evolution, synchronization, and  reconciliation of information states over time. 
By incorporating temporal dimensions into \mls s, we allow logical structures to capture successive stages of knowledge acquisition, revision, and propagation, while merge operations provide formal tools for combining information originating from different sources. 
We introduce generalized modal and temporal operators capable of modeling transitions between paraconsistent, paracomplete, and classical reasoning environments within a unified semantic architecture. 
These extensions reinforce the applicability of \mls s to distributed databases, evolving knowledge systems, and multi-agent environments, where the continuous integration and temporal management of conflicting information are essential for coherent and reliable reasoning.

The remainder of this paper is structured as follows.   
Section \ref{sec.mlms} presents the general machinery of many-logic modal structures. 
In Section~\ref{sec.logics}, we present the matrix logics considered in this paper. These include \letkp\ 
and some well-known two-, three-, and four-valued logics, which are capable of representing subsets of the six scenarios characterized by \letkp. The lattices associated with these logics are down-complete sublattices of $L6$. 
 In Section~\ref{sec.applications}, the main section of this paper, we propose further developments of \mls s, namely: 
 temporalization of \mls s that allow to define temporal operators \textit{past}, \textit{future},  and  \textit{next}; merging of \mls s, that intends to represent the merging of databases; analysis of inconsistencies and certified data in databases.

\vv\section {Many-logic modal structures}   \label{sec.mlms}
In this section, we recall the key definition of many-logic modal structures as presented in \cite[Sect.~2]{Martins2026}. 
We start by introducing some notions related to lattices that will be important in the definition of the main concept of this section.  

\begin{definition} Let  $\mathsf{L}=\langle L,\cdot,{+}\rangle$ be a complete lattice with associated partial order $\leq$. A {\em down-complete sublattice (dc-sublattice)}  of  $\mathsf{L}$ is a complete lattice $\mathsf{L}'=\langle L',\cdot',{+}'\rangle$  with associated partial order $\leq'$ such that:
\begin{enumerate} \setl 
	\item $L' \subseteq L$;
	\item ${\leq}'={\leq} \cap (L' \times L')$, i.e.: for every $x,y \in L'$, $x \leq' y$ iff $x \leq y$.
\end{enumerate}
\mh 
We denote by   $\SubLat(\pair{{L},\cdot,+})$  the set of all \textit{dc-sublattices} of $\pair{{L},\cdot,+}$. 
\end{definition}

\begin{definition}[The down-interpretation]
		 Let  $\mathsf{L}=\langle L,\cdot,{+}\rangle$ be a complete lattice with associated partial order $\leq$, and let  $\mathsf{L}'=\langle L',\cdot',{+}'\rangle$ be a  dc-sublattice  of  $\mathsf{L}$. The {\em down-interpretation} of value $x \in L$ in $\mathsf{L}'$ is defined as follows:
\vv $$x^{\mathsf{L}'} = \bigvee_{\mathsf{L}'} \{y \in L' \ : \ y \leq x\}.$$  
\end{definition}
\begin{definition}[Filtered $\Lg$-lattices]
	Let a \textit{filtered $\Lg$-lattice} be a matrix logic  $ \pair{{\bf L},{\rm D}} $ 
	such that    $\textbf{L}$ is an algebra over $ \Sigma $ with domain $L$ and  $ {\rm D} \subseteq  L $, 
	where the $ \Sigma_L $-reduct  $  \pair{L,\cdot, +} $ is a lattice.
\end{definition}

Consider a propositional \textit{language $\Lg$},  defined over  a signature  $ \Sigma $ containing $ \Sigma_L =  \pair{\land, \lor}$,  and a set of sentential letters~$ Var $. 

\begin{definition} [Many-logic modal  structures, \mls]\label{def.mlms}  
 Let  $\LAT \subseteq \SubLat(\pair{{L},\cdot ,+})$, such that $ \pair{L,\cdot, +} \in \LAT$,   and $\Lg^\Box$ the language obtained by extending $\Lg$ with 
 the unary connective $\Box$.  
   A \textit{many-logic modal structure} $M$ over $\LAT$ and $\Lg^\Box$  is a tuple $\pair{W, R, I,v}$ such that 
$W \neq \emptyset$ and:
	\begin{enumerate} \setl
    \item $\mathbb{L}$ is a set of matrix logics $ \pair{{\bf L}',{\rm D}'} $ over $\Sigma$ such that $ \Sigma_L $-reduct $\pair{{L}',\cdot,+}$ of ${\bf L}'$ belongs to $\LAT $.
\item $I:W\to \mathbb{L}$ assigns a matrix logic $ \pair{{\bf L}_w,{\rm D}_w} $ to each world $w \in W$. By simplicity, and when there is no risk of confusion, we may write $I(w)=L_w$, where $L_w$ is the domain of  ${\bf L}_w$.
\item $R$ is a relation from worlds to worlds, i.e. $R \subseteq W \times W$.
\item   $v$ is a \textit{valuation function} 
 from the set $Var$ of sentential letters to values in $I(w)$, i.e. for each world $w \in W$ and $p\in Var$,  
  $v: W \times Var \longrightarrow L$ such that $v(w,p)\in L_w$.  
	    \end{enumerate}
	\end{definition}

In each world $w$ a valuation $v_w$ (:= $v(w,-)$) 
is extended to all formulas of $\Lg$ based on the  matrix logic $\pair{\textbf{L}_w, {\rm D}_w}$ related to the lattice 
assigned to $w$, given by $I(w)$. 
To extend a valuation $v_w$  to all formulas of 
$\Lg^\Box$ (i.e.~to assign semantic values to formulas $\Box A$) we interpret in a complete sublattice 
$L' \subseteq L$ elements not necessarily in $L'$ according to the down interpretation.

Let $w$ and $w'$ be worlds,  $L_w$ and $L_{w'}$ the corresponding lattices,    
and $v_{w}(A)$ and $v_{w'}(A)$ the values assigned to a formula $A$ in $w$ and $w'$ by a valuation 
$v$. We denote by  $ (v_{w}(A))^{L_{w'}}$ and $ (v_{w'}(A))^{L_{w}}$ the values $v_{w}(A)$ and $v_{w'}(A)$  `seen' from the perspective of the worlds $w'$ and $w$ respectively, according to the down-interpretation.

The semantic value of a  formula $\Box A$ in a world $w$ 
is the infimum of the set of  values of $A$ in all the worlds accessed from $w$ as seen 
from the perspective of the local lattice $L_w$.  
A  clause for the connective $\Box$ is thus  defined as follows:  

\begin{definition}[Semantic clause for $\Box$]\label{def.clause.box}
\mh

 	\begin{center}
		 $v_{{w}}(\Box  A)=\bigwedge\limits_{L_{{w}}}\{(v_{w'}(A))^{L_{{w}}} : w' \in W \land {w} R w'  \}$. 
		\label{clause.box}
	\end{center}
\end{definition}

\mh According to the intuitive interpretation we propose, 
 $v_w(\Box A)=x$  means 
`$A$ has the semantic value $x$ for an agent accessing the information available in the acessible worlds $w$'.  

We say that $A$ holds in $w$ in a structure $M$, denoted by $M, w \vDash A$, if $v_w(A) \in \rm D$. 
It is to be observed that in the framework adopted here, saying that a formula has a specific value  ($v_w(A)=x$)  in a model $M$ and a world $w$  is  generally  
more informative than  saying that it holds  (or does not hold). 
This is due to the multivalued semantics of the  logics we will see in Section~\ref{sec.logics} below, which may involve multiple designated (or non-designated) values. 
Examples of many-logic modal structures based on these logics, 
together with their corresponding lattices, can be found in Section~5 of \cite{Martins2026}.

\vv\section {The logics} \label{sec.logics}

In this paper we consider the following matrix logics: 
\vv\begin{description} \setl 
    
 \begin{multicols}{2}  
    \item  \quad $ \letkp = \langle L6, \{ T, T_0, b \}  \rangle $, 

    \item  \quad  $ \fdetobot = \langle \lw, \{ T_0,  b \}  \rangle $

    \item  \quad  $ \lj = \langle \ls, \{ T,   b \}  \rangle $

    \item  \quad   $ \lpto   = \langle \bw, \{ T_0,   b \}  \rangle $

    \item  \quad  $ \jt   = \langle \bs, \{ T,   b \}  \rangle $

    \item 
$ \ktto = \langle \nw, \{ T_0 \}  \rangle $

    \item $ \lt = \langle \ns, \{ T  \}  \rangle $

    \item $ CL^w   = \langle \cw, \{ T_0  \}  \rangle $

    \item  $ CL^s   = \langle \cs, \{ T  \}  \rangle $

    \item  $ \white CL^s   = \langle \cs, \{ T  \}  \rangle $
 \end{multicols}
\end{description}
\vv\mh  Let the  set $LAT = \{  L6, L4^w, L4^s, B3^w, B3^s, N3^w, N3^s, C2^w, C2^s \}  $
be the collection of algebras underlying the logics   above.
For the semantic characterization of these logics in terms of down-complete sublattices of the lattice \textit{L6}, 
as well as for a detailed presentation of the corresponding lattices, their connection with the scenarios represented by each logic -- which are subsets of the six scenarios represented by \letkp\ -- and their characterization as information-based logics, we refer the reader to Sections~3 and 4 of \cite{Martins2026}.
 Here, we only note that \textit{J3} is the three-valued paraconsistent logic introduced in D'Ottaviano et al. \cite{dottaviano.1970}, 
\textit{\L 3} is the {\L}ukasiewicz' three-valued paracomplete logic, \fdetobot\ and \lpto\ extend 
respectively \fde\ and the logic of paradox \textit{LP} with a material implication and a bottom particle. 
Moreover, \ktto\ extends Kleene's \textit{K3} with a material implication,  and \textit{{\L}J4} may be viewed either as a paraconsistent version of \textit{{\L}3} or as a paracomplete  version of \textit{J3}.

\section {Operations on information structures} \label{sec.applications}

In this section, we introduce operations on \mls s, the evolution-structures that represent their development over time, and the tools for measuring inconsistency.

 Let a \textit{generalized literal} be a formula of the form $p$, $\neg p$, or $ \con p$, for a variable $p$ \footnote{An explanation of the notion of  extended literals $p$, $\neg p$, and $\circ p$ can be found in  \cite{con.rod.2023} and \cite{rodrigues.coniglio.qletf}, where six valued logics of evidence and truth (\lets) are investigated. 
Belnap-Dunn logic works with two primitive notions, positive information and negative information, represented by literals $p$ and $\neg p$. 
The six-valued semantics of \lets\ adds one more primitive notion, reliable information, expressed by $\circ p$.  Once again, we refer the reader in particular  to the paper \cite{rodrigues.coniglio.qletf}, where  these concepts are explained in detail.}. 
A many-logic modal  structure can be viewed as a $W$-family $K_w $ ($w\in W$) of  generalized  
literals defined as the smallest set satisfying the following conditions:
\begin{multicols}{2}
   \begin{itemize}\setl \small 
\item[]  For any propositional variable $p$: 
\item[]
\item[-] if $v_w (p)=T$, then $p,\circ p\in K_w$,  
\item[-] if $v_w (p)=F$, then $\circ p, \neg p\in K_w$,  
\item[-] if $v_w (p)=\bo$, then $p, \neg p\in K_w$,  
{\white\item[-] if $v_w (p)=\nei$, then $\circ p\in K_w$,}  
\item[]
\item[-] if $v_w (p)=\nei$, then $\circ p\in K_w$,  
\item[-] if $v_w (p)=T_0$, then $ p\in K_w$,  
\item[-] if $v_w (p)=F_0$, then $\neg p\in K_w$.
\end{itemize} 
\end{multicols}

\mh An information state is thought of here as the set of  
 formulas that hold in a world $w$, which is the set $K_w$ of generalized literals closed under the notion of 
logical consequence related to the lattice $L_w$.

\subsection {Merging  many-logic modal  structures} 

 The merging of two  compatible  many-logic modal structures can be defined in a natural way. 
 The general idea is as follows: given two structures $M$ and $M'$ and a world $w$ that belongs to both 
 $M$ and $M'$, the merging yields a structure $M''$ in which $w$ contains the generalized literals 
   that belong to $w$ in $M$ and $M''$. 
   This operation, of course, cannot yield a trivial world. 
   Thus, $M$ and $M'$ have to be {compatible} 
 in the sense that in the resulting structure  there cannot be a world 
 with literals $p$, $\neg p$ and $\con p$.

 \begin{definition}
 	Two  many-logic modal  structure $M=\pair{W, R, I,v}$ and $M'=\pair{W', R', I',v'}$ are said to be 
 	{\it incompatible} if for some  $w\in W\cap W'$ there is a sentential letter   $p$ 
   	such that one of the following conditions hold:
 
 	\begin{itemize}\setl
 		\item[-] $v_w(p)=T$ and ($v_w'(p)=b$ or $v_w'(p)=F_0$ or $v_w'(p)=F$)
 		\item[-] $v_w(p)=T_0$ and $v_w'(p)=F$ 
 		\item[-] $v_w(p)=b$ and ($v_w'(p)=T$ or $v_w'(p)=F$)
 		\item[-] $v_w(p)=F_0$ and $v_w'(p)=T$
 		\item[-] $v_w(p)=F$ and ($v_w'(p)=T$ or $v_w'(p)=T_0$ or $v_w'(p)=b$)
 	\end{itemize}
 	When they are not incompatible, they are called {\it  compatible}.
 	 \end{definition}

\mh The definition above ensures that  there is no world in which for a sentential letter   $p$,  
  $\{ p, \neg p, \con p\} \subseteq K_w \cup K'_w$. 
 	Note that two weak worlds (\cw, \bw, \nw, and \lw) are always compatible, since there is no formula $\con p$ in the respective sets $K_w$. 
 Besides, it is not required that a world  $w  \in W\cap W'$ be assigned the same lattice in $M$ and $M'$. 

\mm 
The merging of two compatible structures $M$ and $M'$ results in the 
many-logic modal structure $M//M'$ as defined  below.

 \begin{definition} 	Given two compatible many-logic modal  structure $M=\pair{W, R, I,v}$ and $M'=\pair{W', R', I',v'}$, the 
{\it merging of $M$ with $M'$} is  the many-logic modal structure $M//M' =(W\cup W', R\cup R',\bar I,\bar v)$ where: 
 	 \begin{itemize}\setl
 	\item[-] 
     	$\bar I(w)$ is defined using $I(w)$ and $I'(w)$ according to the Table \ref{tab.I.merging} below.
            \item[-]  $\bar v$ is defined as follows:
 			 \begin{itemize}
     \item[-] For $w\in W\cap W'$: 
 	$\bar v_w(p)$ is defined using $v(p)$ and $v'(p)$ as defined in Table \ref{table.valuation.merging} below, where the $X$ indicates that the merging cannot be done.
 	 	\item[-] For $w \in W\backslash  W'$: 
 	$\bar v_w(p)=v_w(p)$ 
 	\item[-] For $w \in W'\backslash  W$: 
 	$\bar v_w(p)=v'_w(p)$
 \end{itemize}
 	 \end{itemize}

\end{definition}

\begin{table}[H] 
    \centering{\footnotesize
    \begin{tabular}{r|cccccccccc} 
$I(w)\backslash I'(w)$&$L6$&\ls&\lw  & \bs   &\bw   &\ns   &\nw   & \cs   & \cw  \\ \hline
        $L6$ &  $L6$  &$L6$  & $L6$  &$\bs$  &$L6$  &$L6$  &$L6$  & $\cs$ & $L6$ \\ 
         \ls & $L6 $  & \ls  & $L6$  & \bs   & $L6$ & \ls  & $L6$ & \cs   & $L6$ \\ 
         \lw & $L6$   & $L6$ & \lw   & $\bs$ & \bw  & $L6$ & \lw  & $\cs$ & \lw  \\ 
        \bs  &$\bs$   & \bs  & $\bs$ & \bs   & \bs  &$\bs$ & $\bs$& \cs   & \bs \\ 
         \bw & $L6$   & $L6$ & \bw   & \bs   & \bw  & $L6$ & \bw  & \cs   & \bw \\ 
        \ns  & $L6$   & \ls  & $L6$  & \bs   & $L6$ & \ns  &$L6$   & \cs   & $L6$\\ 
         \nw & $L6$   & $L6$ & \lw   & \bs   & \bw  &  $L6$  & \lw  & $\cs$ & \cw \\ 
         \cs & $\cs$  & \cs  & $\cs$ & \cs   & \cs  &\cs   & $\cs$& \cs   & $\cs$\\ 
         \cw & $L6$   & $L6$ & \bw   &\bs    & \bw  & $L6$ &\bw   & $\cs$ &  \bw  \\ 
    \end{tabular}}
    \caption{Merging lattices}
    \label{tab.I.merging}
\end{table}
 \begin{table}[H]{\footnotesize  
$$\begin{array}{r|cccccc}
	v_w\backslash v'_w& T& T_0&\bo & \nei & F_0&F  \\ \hline 
	T & T   & T &X&T&X&X\\
	T_0 &T&T_0&\bo&T_0&\bo &X\\
	\bo &X&\bo&\bo&\bo&\bo&X\\
	\nei  &T&T_0&\bo&\nei &F_0&F\\
	F_0 &X&\bo&\bo&F_0&F_0&F\\
	F &X&X&X&F&F&F
\end{array}
$$}
\caption{Merging valuations}
\label{table.valuation.merging}
\end{table}

  The definition of merging, alternatively, can be done in terms of the  \textit{W}-family $K_w$: 
\begin{itemize}\setl  
   \item[-] For $w\in W\cap W'$:   
 $K_{w_{M//M'}} = K_{w_{M}} \cup K_{w_{M'}} $ 
 	\item[-] For $w \in W\backslash  W'$: 
    $K_{w_{M//M'}} = K_{w_{M}}  $ 
 	\item[-] For $w \in W'\backslash  W$: 
    $K_{w_{M//M'}} = K_{w_{M'}}  $  
\end{itemize}

\mh For each world $w \in  W \cup W'$, the set $K_w$ defines the valuation $\bar v$ according to the  Table \ref{table.valuation.merging}.  
 For worlds in $W \cap W'$, it is not required that the associated lattice be the same.

\medskip

  The rationale for  the Table \ref{tab.I.merging} is as follows. 
Two weak worlds $w$ and $w'$ are always compatible, but when merged the result may be a
 lattice different from $L_w$ and $L_{w'}$. 
 Merging \cw\ with \cw\ yields a \bw: there cannot be any incompleteness in the resulting world, 
 but there may be contradictions. 
 Merging \nw\ with \nw\ yields a \lw: there may be incompleteness as well as contradictions in the resulting world. 
Merging \bw\ with \nw\ yields a \bw: there cannot be any incompleteness in the resulting world, 
 but there may be contradictions. And so on. 

Merging  $L6$ always  yields $L6$, except in the case of \cs\ and \bs. In the former, for every $p$, 
 either both $\con p$ and $p$  or both $\con p$ and $\neg p$ belong to $w$; in the case of \bs, 
 either $\con p$ and $p$,  or $\con p$ and $\neg p$, or $p$ and $\neg p$ belong to $w$. So, clearly the resulting 
 world must be \cs\ and \bs, respectively. It is unlikely, however, that worlds \bs\ and \cs\ will be compatible 
 with $L6$. 

Another case that is worth mentioning is the merging of \ns\ with \lw. The resulting world may have incompleteness 
from both \ns\ and \lw, and contradictions from \lw, as well as certified information from \ns\ and non-certified 
information from \lw. Thus, such a merging yields $L6$.   
 

\newpage

\begin{example}  Consider the structures $M_1$ and $M_2$: 
 \begin{multicols}{2}
\scriptsize 
\mh $M_1$

\mh $I(w_1)= L6, I(w_2)=\lw, I(w_3)= \bs $

\mh $v_{w_1}(p)=T, v_{w_2}(p)=T_0, v_{w_3}(p)=T $

\mh $R_1=\{\pair{w_1,w_2}, \pair{w_1,w_3}, \pair{w_1,w_1},\\ \pair{w_2,w_2},\pair{w_3,w_3}  \}$
 \input{ex/example6.1b} 

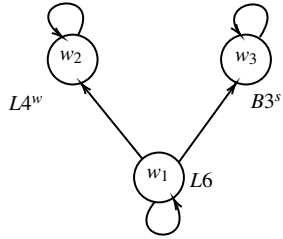
\mmmm\m 

$M_2$

$I(w_1)=\lw , I(w_2)=\cw $

$v_{w_1}(p)=T_0, v_{w_2}(p)=F_0    $

$R_2=\{\pair{w_1,w_2}, \pair{w_2,w_1}, \pair{w_1,w_1},\pair{w_2,w_2} \}$

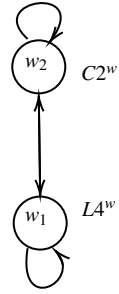
\input{ex/ex.6.1.OK} 

\end{multicols}


The merging of $M_1$ and $M_2$ yields the structure $M_1//M_2$ below: 

\medskip

$\bar I(w_1)=L6, \bar  I(w_2)=\bw , \bar  I(w_3)=\bs $

$\bar v_{w_1}(p)=T_0, \bar  v_{w_2}(p)=\bo, \bar  v_{w_3}(p)=T  $

$R_1\cup R_2 = \{\pair{w_1,w_2}, \pair{w_1,w_3}, \pair{w_2,w_1}, \pair{w_1,w_1},\pair{w_2,w_2},\pair{w_3,w_3}  \}$

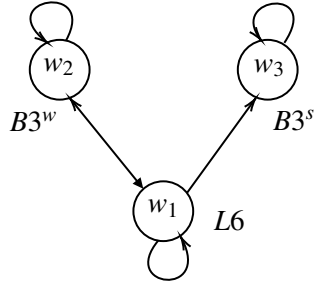
\begin{center}
    \input{ex/example6.1c} 
    \end{center}

\end{example}

\mh When applying the merging operation in a real-world situation, names may be assigned to databases so as to indicate which of them are to be merged. Databases sharing the same name are intended to be merged into a single database.
In the example above, the idea was to keep the \bs-database of $M_1$ unchanged while merging the $L6$-database of $M_1$ with the \lw-database of $M_2$, and the \lw-database of $M_1$ with the \cw-database of $M_2$. The names assigned to the worlds reflect the intended merging process: databases given the same name are merged together.

\begin{proposition} \ 
 
 \mh	(i) Let $M=\pair{W, R, I, v}$ and $M'=\pair{W, R, I, v'}$ be  two  many-logic modal structures. Then
	for any $w\in W$ and any formula $A$, $v_w (A)\in\{b,T_0,T\}$ or  $v'_w (A)\in\{b,T_0,T\}$ imply $\bar{v}_w (A)\in\{b,T_0,T\}$.

 \mh (ii) The merging operation  is commutative and associative, but not idempotent. 
\begin{proof} Regarding item (i), it is easy to see that if $A$ is assigned designated value by  $v_w$ or $v'_w$, which is a value  represented by a snapshot $(z_1, z_2, z_3)$ such that $z_1 = 1$, the value assigned to $A$ by  $\bar{v}_w$ will also be designated. As for (ii), it is clear to see that   
commutativity follows from the fact that $\cup$ and $\cap$ are commutative and 
the Table \ref{tab.I.merging} is commutative.  Associativity    follows from the fact that $\cup$ and $\cap$ are  associative and the Table \ref{tab.I.merging} in each case the resulting value depends on $K_w \cup K'_w$. To see that idempotency does not hold, consider
the merging of two \cw-worlds, which yields a \bw-world. \end{proof}
 \end{proposition}

\subsection {Temporalising  many-logic modal structures}

In Example 5.4 of \cite{Martins2026}, we showed how the temporal evolution of  information states over time 
  can be represented using only the tools already available in the definition of \mls s.
However, the temporal evolution of information states may also be represented by means of what we call 
here \textit{T-evolution structures}.

Consider now the language $\Lg^{\mathsf{T}}$, obtained by extending the language $\Lg^\Box$  with the 
unary operators \pp\ and \ff\ with the following meanings: 
\begin{itemize} \setl 
     \item[]  $\pp (A)$:  \textit{the semantic value of $A$  in the past}, 

 $\ff (A)$:  \textit{the semantic value of $A$ in the future}.  

 \end{itemize}

\mh The semantic values of $A$ in the past and in the future are given, respectively, by the infimum of the values that $A$ assumes in the past and by the infimum of the values that $A$ assumes in the future.

\mm 
\begin{definition}  \m \mh A {\it pointed many-logic modal structure} is a pair $(M,w_0)$, where  $M = \pair{W, R, I, v}$ is a many-logic modal  structure and $w_0\in W$.  $w_0$ is called the reference world of $M$. 
\end{definition}

A pointed many-logic modal structure is a many-logic modal structure together with a distinguished world. 
 {The purpose of introducing pointed structures is to make explicit the world at which formulas are evaluated in 
 $T$-evolution structures.} 
  Let $PAT$ be a set of pointed many-logic modal structures $\pair{M,w_0}$.

 \begin{definition}[$T$-evolution structures]  \label{def.t.evol} 
 \mh A \textit{$T$-evolution structure} ${M^{\mathsf{T}}}$ is a triple $\pair{T,<, g}$, 
    where  $\pair{T,<}$ is a  flow of time (that is, $T$ is a non-empty set and $<$ is a partial order on $T$), 
    and $g:T \to PAT$ is a function mapping each time point in $T$ to a structure $M$ in $PAT$.

  \m\mh The value $ \tau $ of a formula $A$ in a structure $M^{\mathsf{T}}$ at a time point $t\in T$ is determined by a mapping $\tau(A,t) \to L6$, defined as follows:
	\begin{enumerate}\setl
	             \item  $\tau (A, t ) =v(A, {w_0})$, where $w_0$ is the reference world of $g(t)$,
		\item  $\tau ({\pp A},t)= \bigwedge  \{\tau ({A},s) :  s<t\}$,
\item  $\tau ({\ff A},t)= \bigwedge  \{\tau ({A},s) : t<s\}$,

\item $\tau  ({\neg A},t)=\neg  \tau({A}, t)$,
 
	\item  $\tau(A \wedge B,t)= \tau(A, t)  \wedge \tau(B, t)  $,

\item $\tau(A \lor B,t)= \tau(A, t)  \lor \tau(B, t)  $,
 
\item  $\tau(A \to B,t)= \tau(A, t)  \to \tau(B, t)  $.  
    
    \end{enumerate}	
 
\mh A \textit{path} in a \textit{T}-evolution structure $M^{\mathsf{T}}$ is a subset $T' \subseteq T$ such that the elements of $T'$ 
are linearly ordered with respect to the ordering defined on $T$.

 \end{definition}

Some remarks are in order here:\\  (i) If $(T,<)$ is linear, discrete, and without end point,  a next operator  \xx\    
 can be added with the following semantic clause: 
\begin{center}
     $\tau ({\xx A},t)= \tau (A,t')$, where $t'$ is the time point immediately after $t$. 
\end{center}

\mh (ii) A path can be any sequence of worlds linearly ordered, and a path can be contained in another path.  
Paths will be used below  to define increasing and decreasing of consistency  and certification 
(Definition~\ref{def.inc.decr}).

 \mh (iii) The  order $<$ applies only to the
 reference worlds of the structures in $PAT$. 
 In each structure $M$, the reference world $w_0$ can intuitively be interpreted   as a sort of 
 `central  database' that accesses every other database. 

 \mh (iv)  The value $\tau(\Box A, t)$ in a structure $M^{\mathsf{T}}$ is  the value of $\Box A$ at the reference world of the structure assigned to $t$ in $M^{\mathsf{T}}$. 
 The values of $({\ff{\Box A}},t)$ and $({\pp{\Box A}},t)$ are given by the infimum of the values of 
 ${\Box A}$ after  and before $t$ respectively.

\mmm

\begin{example} Consider the  set  
 $PAT=\{ \pair{M_1, w_1} , \pair{M_2, w_3}, \pair{M_3, w_5}, \pair{M_4, w_7}  \}$ of pointed many-logic structures.  

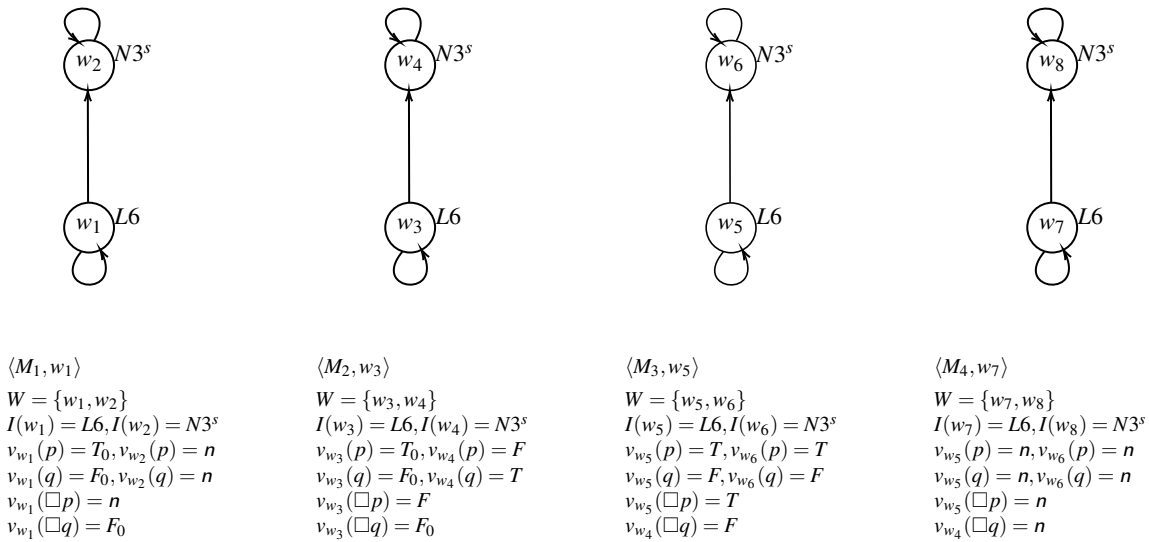
\mmm 

\noindent   {\footnotesize \begin{tabular}{c c c c}
\input{ex/ex10a}  &
  \input{ex/ex10b} &
  \input{ex/ex10c} &
  \input{ex/ex10d}
  \end{tabular}}


\noindent \begin{multicols}{4}
\scriptsize 
\mh $\pair{M_1, w_1} $

\smallskip

\mh $W=\{ w_1, w_2\}$

\mh  $I(w_1)= L6 , I(w_2)=\ns $

\mh $v_{w_1}(p)=T_0, v_{w_2}(p)=\nei $ 

\mh $v_{w_1}(q)=F_0, v_{w_2}(q)=\nei $ 

\mh $ v_{w_1}(\Box p)=\nei  $ 

\mh $ v_{w_1}(\Box q)=F_0  $ 


\mh $\pair{M_2, w_3} $

\smallskip

\mh $W=\{ w_3, w_4\}$

\mh  $I(w_3)=  L6, I(w_4)=\ns  $

\mh $v_{w_3}(p)=T_0, v_{w_4}(p)=F $ 

\mh $v_{w_3}(q)=F_0, v_{w_4}(q)=T $

\mh $ v_{w_3}(\Box p)=F  $ 

\mh $ v_{w_3}(\Box q)=F_0  $ 


\mh $\pair{M_3, w_5} $

\smallskip

\mh $W=\{ w_5, w_6\}$

\mh  $I(w_5)=  L6, I(w_6)=\ns  $

\mh $v_{w_5}(p)=T, v_{w_6}(p)={T}$ 

\mh $v_{w_5}(q)=F, v_{w_6}(q)= {F} $

\mh $v_{w_5}(\Box p)=T $ 

\mh $ v_{w_4}(\Box q)=F $ 



\mh $\pair{M_4, w_7} $

\smallskip

\mh $W=\{ w_7, w_8 \}$

\mh  $I(w_7)=  L6, I(w_8)=\ns  $

\mh $v_{w_5}(p)=\nei , v_{w_6}(p)=\nei $ 

\mh $v_{w_5}(q)=\nei , v_{w_6}(q)= \nei $

\mh $v_{w_5}(\Box p)=\nei  $ 

\mh $ v_{w_4}(\Box q)=\nei  $

\end{multicols}

\mm\mh  Consider the  $T$-evolution structure $M^{\mathsf{T}}$: 

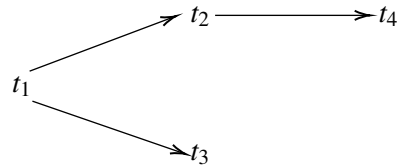
\m 
\begin{figure}[H]
    \centering
\input{ex/ex.T}
    
\end{figure}

\m 

\mh  $T=\{ 1,2,3,4\}$, $g(t_1)=M_1, \; g(t_2)=M_2, \; g(t_3)=M_3, \; g(t_4)=M_4 $.

\begin{multicols}{4}
    \small
    \mh   $\tau(\pp (\Box p,t_1))=\nei $ 

\mh   $\tau(\pp (\Box p,t_2))=F$ 

\mh $\tau(\pp (\Box p,t_3))=F$ 

\mh $\tau(\pp (\Box p,t_4))=F$ 


\mh  $\tau(\ff (\Box p,t_1))=F$ 

\mh   $\tau(\ff  (\Box p,t_2))=F$ 

\mh $\tau(\ff  (\Box p,t_3))=T$

\mh $\tau(\ff  (\Box p,t_4))=\nei $ 


\mh  $\tau(\pp (\Box q,t_1))=F$ 

\mh   $\tau(\pp (\Box q,t_2))=F_0$ 

\mh $\tau(\pp (\Box q,t_3))=F$

\mh $\tau(\pp (\Box q,t_4))=\nei $ 


\mh  $\tau(\ff (\Box q,t_1))=\nei $ 

\mh   $\tau(\ff  (\Box q,t_2))=F$ 

\mh $\tau(\ff  (\Box q,t_3))=\nei $ 

\mh $\tau(\ff  (\Box q,t_4))=F$ 
\end{multicols}

 \end{example}

    The intuitive interpretation can be as follows. The path  $t_1,t_3$ represents a circumstance in which all the 
information available in the reference world $w_1$ of the structure $M_1$  has been refined. 
The resulting database, which corresponds to reference world $w_5$ of the structure $M_3$, 
has  certified information that is the same in the worlds $w_5$ and $w_6$.    
 The uncertainty has decreased. 

The path  $t_1,t_2,t_4$, on the other hand, represents an inverse  circumstance in which the uncertainty has increased.  
 Different information came from the database represented   by the worlds $w_4$ in the stages $t_2$.  
 This uncertainty is represented in the stage $t_4$ by the value \nei\ in both worlds.

\vv\subsection {Measuring inconsistency}

 The quality of information in a database depends on both its inconsistencies and the reliability of the information it contains.   The logic \letkp\      distinguishes reliable from unreliable information by means of the connective \con, which is intuitively understood  
 as a form of data certification. 
  
 In this section, we define measures of inconsistency and certification within a structure  $M$, 
 based  on the number of occurrences of  pairs of sentences $p, \neg p$ as well as instances of $\con p$ in the respective $W$-family $K_w$ (cf.~\cite{CostaM17,DianaManuel18,Grant06,grant18}). 
  Although these measures are related, 
they remain independent in the sense that, except in the logics \jt\ and \lj,  it is possible   to have no contradiction regarding $p$,   yet $\con p$ does not hold.
We further extend these concepts to \textit{T}-evolution structures and introduce the notions of consistency improvement and certification improvement over time.

 \begin{definition} \label{def.measurin.inconsistency} \  
  
Let $M = \pair{W, R, I, v}$ be a  many-logic modal structure where $|W|$ is finite.

\m\mh (i) The {\em inconsistencies at $w \in W$}, $inc(w)$, is the number of variables $p$ such that $p,\neg p$ are in the 
database, that is:

\m  $inc(w)= |\{p: v_w(p)=\bo \}|,$ or alternatively:
$ inc(w)= |\{p:  p \in K_w \text{ and } \neg p \in K_w \}|.$

\mm \mh (ii) The {\em certificated data at $w$}, $\cert(w)$, is  the number of variables $p$ such that $\circ p$ is in the 
database, that is: 

\m  $\cert(w)=|\{p: v_w(p)=T \text{ or }v_w(p)=F\}|,$  or alternatively:
$\cert(w)=|\{p: \con  p \in K_w \}|.$

\m \mh (iii) The {\em inconsistencies in $M$}, $inc(M)$, is the sum of the inconsistencies in the worlds $w$ 
such that $w \in W$,  that is:    

\vv $$inc({M})=    {\sum_{w\in W}\inc(w)}$$
\vv

\m\mh (iv) The {\em certificated data in $M$}, $\cert(M)$,  is the sum of the certified data in the worlds $w$ 
such that $w \in W$,  that is:    

\vv \vv 
$$cert({M})={\sum_{w\in W}\cert(w)}$$
\vv \vv

\mh (v) The {\em relative inconsistency of $M$}, $rinc(M)$, is the sum of the inconsistencies in $M$ divided by 
the cardinality of $W$,   that is:    

\vv \vv 
$$\rinc(M)= \frac{ \inc(M) }{|W|}$$
\vv \vv

\mh (vi) The {\em relative certified data in $M$}, $rcert(M)$,  is the sum of the certified data in $M$ divided by 
the cardinality of $W$,   that is:    

\vv \vv 

$$rcert({M})=\frac{ \cert(M) }{|W|}$$
\vv \vv

 \end{definition}

In particular cases data do not have the same relevance and we must consider this in the previous definition. Let $\Lambda\subset N$ be a scale of relevance for data and $\mu :Var\to \Lambda $.
\begin{definition}
   \m\mh  
Let $M = \pair{W, R, I, v}$ be a  many-logic modal structure where $|W|$ is finite and $\mu :Var\to \Lambda\subset N$. 
\\
\m\mh (i) The {\em inconsistencies at $w \in W$ modulo $\mu$}, $inc^\mu(w)$, is the $\mu$-weighted sum of the number of variables $p$ such that $p,\neg p$ are in the 
database, that is:    
$$inc^\mu(w)= \sum_{ v_w(p)=\bo } \mu(p),$$
\\
\m\mh (ii) The {\em certificated data at $w$ modulo $\mu$}, $\cert^\mu(w)$, is  the $\mu$-weighted sum of the  number of variables $p$ such that $\circ p$ is in the 
database, that is: 
$$\cert^\mu(w)=\sum_{ v_w(p)=T \text{ or }v_w(p)=F } \mu(p),$$

\end{definition} 

The other definitions can also be extended with weights for the proposition symbols.

\mh
  The point of the notions of relative inconsistency and certification  is to obtain a normalized result   proportional to the size of $W$.   
Note that the Definition~\ref{def.measurin.inconsistency} above imposes  no restrictions on the cardinality of the  sets 
$\{p: v_w(p)=\bo\}$,  $\{p: v_w(p)=T\}$, and   $\{p: v_w(p)=F\}$. However, if they are not finite,  
the values of $rcert_{M}$ and $rinc_{M}$ will be $\omega$.

\begin{proposition}  Let $M=\pair{W, R, I,v}$ and $M'=\pair{W', R', I',v'}$ be two compatible many-logic modal  structure. Then:  
 
\mm 
(i)   $cert(M)\leq cert(M//M')$, 

\m (ii) $inc(M)\leq inc(M//M')$, 

\m 
(iii) If $W\cap W'=\emptyset$ then $inc(M//M')= inc(M)+inc(M')$ and $cert(M//M')= cert(M)+cert(M')$. 
 \end{proposition}
\begin{proof}
  Item (i)  follows from the fact that  $|\{p: \con  p \in K_w \}|  \leq |\{p: \con  p \in K_w \cup K'_w \}|$. 
  The remaining cases are left to the reader. 

\end{proof}

Below we extend these ideas to $T$-evolution structures.

\begin{definition} \label{def.Relative.inconsistency}  
  Let $(T,<)$ be a flow of time and  $M^{\mathsf{T}}$ a $T$-evolution structure, where $|W_t|$ is  finite for every $t\in T$.

 (i) The {\em relative inconsistent measure at t} is the function $\rinc_{M^{\mathsf{T}}} :  T \to [0,1]$ defined as

$$rinc_{M^{\mathsf{T}}}(t)= rinc(g(t)) $$



(ii) The {\em certification  ratio} is the function  $rcert_{M^{\mathsf{T}}} : T \to [0,1]$

$$rcert_{M^{\mathsf{T}}}(t)= rinc(g(t)) $$

\end{definition}

\mmm 
Some useful properties of $T$-evolution structures may be defined.

\newpage

\begin{definition} \label{def.inc.decr} Let $T'$ be a path in a $T$-evolution structure. We say that:

\begin{itemize} \setl

   \item[(i)]   \emph{Consistency increases} along $T'$ if and only if
 for any $i,j\in T'$, $i<j$ implies  $rinc_{M^{\mathsf{T}}}(i) > rinc_{M^{\mathsf{T}}}(j)$, 


\item[(ii)]   \emph{Consistency decreases} along $T'$ if and only if
for any $i,j\in T'$, $i<j$ implies  $rinc_{M^{\mathsf{T}}}(i) < rinc_{M^{\mathsf{T}}}(j)$,

     \item[(iii)]   \emph{Certification increases} along $T'$ if and only if
for any $i,j\in T'$, $i<j$ implies  $rcert_{M^{\mathsf{T}}}(i) < rcert_{M^{\mathsf{T}}}(j)$, 


\item[(iv)] \emph{Certification decreases} along $T'$ if and only if
 for any $i,j\in T'$, $i<j$ implies  $rcert_{M^{\mathsf{T}}}(i) > rcert_{M^{\mathsf{T}}}(j)$. 


\end{itemize}


\begin{itemize} \setl 
\item[] We say that a $T$-evolution structure: 
\item[(v)]    \emph{Improves consistency}  if and only if  
  for any $i,j\in \T$, $i<j$ implies  $\rinc(i) > \rinc(j)$, 

\item[(vi)]  \emph{Improves certification}    if and only if
  for any $i,j\in \T$, $i<j$ implies  $rcert_{M^{\mathsf{T}}}(i) < rcert_{M^{\mathsf{T}}}(j)$.

\end{itemize}

\end{definition}


\m\mh  In the example below we illustrate the measures  of inconsistency, certification, 
relative inconsistency and certified data in a \textit{T}-evolution structure.

 \begin{example} 
 Consider the set $PAT=\{ \pair{M_1,w_1},\pair{M_2,w_3}, \pair{M_3,w_5} \}$

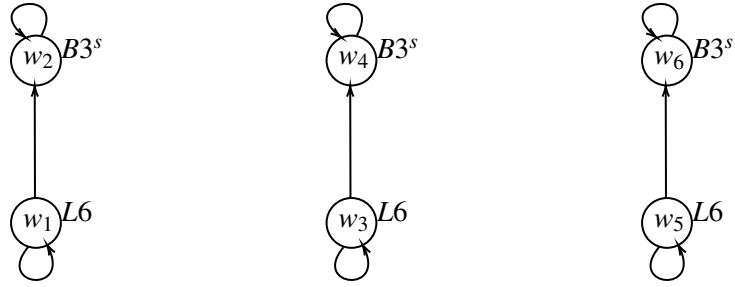
\begin{center}
\begin{tabular}{c c c}
    \input{ex/ex.last.1}&\input{ex/ex.last.2}&\input{ex/ex.last.3}
\end{tabular}
\end{center}

{\footnotesize 
\begin{multicols}{3}

$M_1$
 
\smallskip

$I(w_1)= L6, I(w_2)= \bs  $ 

$K_{w_1} = \{ \con p, p, q\} $ 

$ K_{w_2} = \{   p, \neg p, q, \neg q\}$

$inc(w_1)= 0, inc(w_2)=2$ 

$cert(w_1)= 1, cert(w_2)=0 $  

$inc(M_1)=2, cert(M_1)=1$

$rinc(M_1)=1, rcert(M_1)=0.5$

$M_2$

\smallskip

$ I(w_3)= L6, I(w_4)= \bs  $

$K_{w_3} = \{  p, \neg p, q \} $  

$K_{w_4} = \{   p, \neg p,   q, \neg q\}$

$inc(w_3)=1, inc(w_4)=2  $

$ cert(w_3)=0, cert(w_4)=0  $

$inc(M_2)=3, cert(M_2)=0$

$rinc(M_2)=1.5, rcert(M_2)=0$

$M_3$

\smallskip

$ I(w_3)= L6, I(w_4)= \bs $

$K_{w_1} = \{ \con p, p, \con q, q\} $ 

$ K_{w_2} = \{   \con p,  p, \con q,  q\}$

$inc(w_1)= 0, inc(w_2)=0$ 

$cert(w_1)= 2, cert(w_2)=2 $  

$inc(M_3)=0, cert(M_3)=4$

$rinc(M_3)=0, rcert(M_3)=2$

\end{multicols}
}

\m\mh Now consider the \textit{T}-evolution structure $M^{\mathsf{T}}$ where $T=\{ 1,2,3\}$, $g(t_1)=M_1, \; g(t_2)=M_2, \; g(t_3)=M_3.$

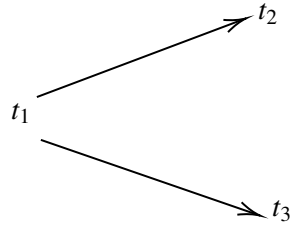
\begin{figure}[H]
    \centering
    \input{ex/ex.tree.last}
\end{figure}

\begin{multicols}{2}
    
$rinc(t_1)=1 $  

$rinc(t_2)=1.5 $ 

$ rinc(t_3)=0  $  

$rcert(t_1)=0.5 $ 

$ rcert(t_2)=0 $ 

$ rcert(t_3)=2.  $  
\end{multicols}
 
\m \mh
Consistency increases and certification decreases along the path $\T' = \{t_1, t_2\}$, while consistency decreases and certification increases along the path $\T'' = \{t_1, t_3\}$. However, with regard to the \textit{T}-structure $M^{\mathsf{T}}$, neither consistency nor certification increases or decreases.

 \end{example}

 \subsection*{Final remarks}
In this paper, we discussed mechanisms for database merging and temporal evolution, emphasizing how informational refinement may be modeled through changes in the underlying lattices. In this setting, inconsistency measures play an important role in characterizing the informational state of databases and in guiding their evolution toward more consistent and certified forms of knowledge.

Future research includes the investigation of refinement processes capable of capturing da\-ta\-ba\-se evolution in terms of certification and inconsistency levels through (T)-evolution structures.
Another relevant direction concerns the extension of the lattice (L6) by iteratively introducing new maximal and minimal elements, allowing the representation of finer degrees of certification. Determining the logics generated by the corresponding filters in these extended structures remains an open and promising problem.

\medskip

\noindent {\bf Acknowledgments.}
This work is supported by CIDMA (https://ror.org/05pm2mw36) through the Portuguese Foundation for Science
and Technology (FCT, https://ror.org/00snfqn58), under grant UID/04106/2025 (https://doi.org/10.54499/UID/04106/2025) and by Project Banksy with reference COMPETE2030-FEDER-00892000. The second author acknowledges support from the National Council for Scientific
and Technological Development (CNPq), Brazil, research grant 307889/2025-4.

  \bibliographystyle{eptcs}        
\bibliography{refs.main}

\end{document}

%% file: ex/example6.1b.tex
\tikzset{every picture/.style={line width=0.75pt}} 

\begin{tikzpicture}[x=0.75pt,y=0.75pt,yscale=-0.5,xscale=0.5]

\draw   (81,1493.88) .. controls (81,1480.07) and (92.19,1468.88) .. (106,1468.88) .. controls (119.81,1468.88) and (131,1480.07) .. (131,1493.88) .. controls (131,1507.69) and (119.81,1518.88) .. (106,1518.88) .. controls (92.19,1518.88) and (81,1507.69) .. (81,1493.88) -- cycle ;
\draw   (255,1493.88) .. controls (255,1480.07) and (266.19,1468.88) .. (280,1468.88) .. controls (293.81,1468.88) and (305,1480.07) .. (305,1493.88) .. controls (305,1507.69) and (293.81,1518.88) .. (280,1518.88) .. controls (266.19,1518.88) and (255,1507.69) .. (255,1493.88) -- cycle ;
\draw   (167.67,1612.21) .. controls (167.67,1598.4) and (178.86,1587.21) .. (192.67,1587.21) .. controls (206.47,1587.21) and (217.67,1598.4) .. (217.67,1612.21) .. controls (217.67,1626.02) and (206.47,1637.21) .. (192.67,1637.21) .. controls (178.86,1637.21) and (167.67,1626.02) .. (167.67,1612.21) -- cycle ;
\draw    (115.93,1519.43) -- (175.94,1593.44) ;
\draw [shift={(114.67,1517.88)}, rotate = 50.96] [color={rgb, 255:red, 0; green, 0; blue, 0 }  ][line width=0.75]    (10.93,-3.29) .. controls (6.95,-1.4) and (3.31,-0.3) .. (0,0) .. controls (3.31,0.3) and (6.95,1.4) .. (10.93,3.29)   ;
\draw    (268.64,1519.54) -- (212.6,1596.77) ;
\draw [shift={(269.82,1517.92)}, rotate = 125.96] [color={rgb, 255:red, 0; green, 0; blue, 0 }  ][line width=0.75]    (10.93,-3.29) .. controls (6.95,-1.4) and (3.31,-0.3) .. (0,0) .. controls (3.31,0.3) and (6.95,1.4) .. (10.93,3.29)   ;
\draw    (188.56,1636.62) .. controls (154.92,1676.96) and (233.64,1684.55) .. (208.17,1635.33) ;
\draw [shift={(207.36,1633.82)}, rotate = 61.28] [color={rgb, 255:red, 0; green, 0; blue, 0 }  ][line width=0.75]    (10.93,-3.29) .. controls (6.95,-1.4) and (3.31,-0.3) .. (0,0) .. controls (3.31,0.3) and (6.95,1.4) .. (10.93,3.29)   ;
\draw    (111.36,1469.33) .. controls (138.3,1419.86) and (49.71,1434.4) .. (94.37,1470.24) ;
\draw [shift={(95.76,1471.33)}, rotate = 217.44] [color={rgb, 255:red, 0; green, 0; blue, 0 }  ][line width=0.75]    (10.93,-3.29) .. controls (6.95,-1.4) and (3.31,-0.3) .. (0,0) .. controls (3.31,0.3) and (6.95,1.4) .. (10.93,3.29)   ;
\draw    (293.91,1471.36) .. controls (325.59,1423.84) and (247.5,1432.52) .. (273.48,1467.85) ;
\draw [shift={(274.3,1468.93)}, rotate = 231.92] [color={rgb, 255:red, 0; green, 0; blue, 0 }  ][line width=0.75]    (10.93,-3.29) .. controls (6.95,-1.4) and (3.31,-0.3) .. (0,0) .. controls (3.31,0.3) and (6.95,1.4) .. (10.93,3.29)   ;

\draw (90,1481.73) node [anchor=north west][inner sep=0.75pt]   [align=left] {$w_2$};
\draw (266,1482.73) node [anchor=north west][inner sep=0.75pt]   [align=left] {$w_3$};
\draw (177.33,1600.06) node [anchor=north west][inner sep=0.75pt]   [align=left] {$w_1$};
\draw (133,1537.39) node   [align=left] {\begin{minipage}[lt]{68pt}\setlength\topsep{0pt}
$\lw$
\end{minipage}};
\draw (282.33,1525.06) node [anchor=north west][inner sep=0.75pt]   [align=left] {$\bs$};
\draw (314.67,1615.06) node   [align=left] {\begin{minipage}[lt]{68pt}\setlength\topsep{0pt}
$L6$
\end{minipage}};

\end{tikzpicture}

%% file: ex/ex.6.1.OK.tex
\tikzset{every picture/.style={line width=0.75pt}} 

\begin{tikzpicture}[x=0.75pt,y=0.75pt,yscale=-0.7,xscale=0.7]

\draw   (13,171) .. controls (13,160.51) and (21.51,152) .. (32,152) .. controls (42.49,152) and (51,160.51) .. (51,171) .. controls (51,181.49) and (42.49,190) .. (32,190) .. controls (21.51,190) and (13,181.49) .. (13,171) -- cycle ;
\draw    (31,80) -- (31.97,149) ;
\draw [shift={(32,151)}, rotate = 269.19] [color={rgb, 255:red, 0; green, 0; blue, 0 }  ][line width=0.75]    (10.93,-3.29) .. controls (6.95,-1.4) and (3.31,-0.3) .. (0,0) .. controls (3.31,0.3) and (6.95,1.4) .. (10.93,3.29)   ;
\draw    (32,151) -- (31.03,82) ;
\draw [shift={(31,80)}, rotate = 89.19] [color={rgb, 255:red, 0; green, 0; blue, 0 }  ][line width=0.75]    (10.93,-3.29) .. controls (6.95,-1.4) and (3.31,-0.3) .. (0,0) .. controls (3.31,0.3) and (6.95,1.4) .. (10.93,3.29)   ;
\draw    (23,39) .. controls (-8.68,10.29) and (69.41,-0.78) .. (39.93,39.76) ;
\draw [shift={(39,41)}, rotate = 307.3] [color={rgb, 255:red, 0; green, 0; blue, 0 }  ][line width=0.75]    (10.93,-3.29) .. controls (6.95,-1.4) and (3.31,-0.3) .. (0,0) .. controls (3.31,0.3) and (6.95,1.4) .. (10.93,3.29)   ;
\draw    (23,189) .. controls (15.16,238) and (62.06,213.05) .. (44.18,192.27) ;
\draw [shift={(43,191)}, rotate = 45] [color={rgb, 255:red, 0; green, 0; blue, 0 }  ][line width=0.75]    (10.93,-3.29) .. controls (6.95,-1.4) and (3.31,-0.3) .. (0,0) .. controls (3.31,0.3) and (6.95,1.4) .. (10.93,3.29)   ;
\draw   (12,59) .. controls (12,48.51) and (20.51,40) .. (31,40) .. controls (41.49,40) and (50,48.51) .. (50,59) .. controls (50,69.49) and (41.49,78) .. (31,78) .. controls (20.51,78) and (12,69.49) .. (12,59) -- cycle ;

\draw (60,56) node [anchor=north west][inner sep=0.75pt]   [align=left] {\cw};
\draw (60,153) node [anchor=north west][inner sep=0.75pt]   [align=left] {\lw};
\draw (19,161) node [anchor=north west][inner sep=0.75pt]   [align=left] {$w_1$};
\draw (19,50) node [anchor=north west][inner sep=0.75pt]   [align=left] {$w_2$};

\end{tikzpicture}

%% file: ex/example6.1c.tex
\tikzset{every picture/.style={line width=0.75pt}} 

\begin{tikzpicture}[x=0.75pt,y=0.75pt,yscale=-0.6,xscale=0.6]

\draw   (81,1493.88) .. controls (81,1480.07) and (92.19,1468.88) .. (106,1468.88) .. controls (119.81,1468.88) and (131,1480.07) .. (131,1493.88) .. controls (131,1507.69) and (119.81,1518.88) .. (106,1518.88) .. controls (92.19,1518.88) and (81,1507.69) .. (81,1493.88) -- cycle ;
\draw   (255,1493.88) .. controls (255,1480.07) and (266.19,1468.88) .. (280,1468.88) .. controls (293.81,1468.88) and (305,1480.07) .. (305,1493.88) .. controls (305,1507.69) and (293.81,1518.88) .. (280,1518.88) .. controls (266.19,1518.88) and (255,1507.69) .. (255,1493.88) -- cycle ;
\draw   (167.67,1612.21) .. controls (167.67,1598.4) and (178.86,1587.21) .. (192.67,1587.21) .. controls (206.47,1587.21) and (217.67,1598.4) .. (217.67,1612.21) .. controls (217.67,1626.02) and (206.47,1637.21) .. (192.67,1637.21) .. controls (178.86,1637.21) and (167.67,1626.02) .. (167.67,1612.21) -- cycle ;

\draw    (115.93,1519.43) -- (174.05,1591.11) ;
\draw [shift={(175.94,1593.44)}, rotate = 230.96] [fill={rgb, 255:red, 0; green, 0; blue, 0 }  ][line width=0.08]  [draw opacity=0] (8.93,-4.29) -- (0,0) -- (8.93,4.29) -- cycle    ;
\draw [shift={(114.67,1517.88)}, rotate = 50.96] [color={rgb, 255:red, 0; green, 0; blue, 0 }  ][line width=0.75]    (10.93,-3.29) .. controls (6.95,-1.4) and (3.31,-0.3) .. (0,0) .. controls (3.31,0.3) and (6.95,1.4) .. (10.93,3.29)   ;

\draw    (268.64,1519.54) -- (212.6,1596.77) ;
\draw [shift={(269.82,1517.92)}, rotate = 125.96] [color={rgb, 255:red, 0; green, 0; blue, 0 }  ][line width=0.75]    (10.93,-3.29) .. controls (6.95,-1.4) and (3.31,-0.3) .. (0,0) .. controls (3.31,0.3) and (6.95,1.4) .. (10.93,3.29)   ;

\draw    (188.56,1636.62) .. controls (154.92,1676.96) and (233.64,1684.55) .. (208.17,1635.33) ;
\draw [shift={(207.36,1633.82)}, rotate = 61.28] [color={rgb, 255:red, 0; green, 0; blue, 0 }  ][line width=0.75]    (10.93,-3.29) .. controls (6.95,-1.4) and (3.31,-0.3) .. (0,0) .. controls (3.31,0.3) and (6.95,1.4) .. (10.93,3.29)   ;
\draw    (111.36,1469.33) .. controls (138.3,1419.86) and (49.71,1434.4) .. (94.37,1470.24) ;
\draw [shift={(95.76,1471.33)}, rotate = 217.44] [color={rgb, 255:red, 0; green, 0; blue, 0 }  ][line width=0.75]    (10.93,-3.29) .. controls (6.95,-1.4) and (3.31,-0.3) .. (0,0) .. controls (3.31,0.3) and (6.95,1.4) .. (10.93,3.29)   ;
\draw    (293.91,1471.36) .. controls (325.59,1423.84) and (247.5,1432.52) .. (273.48,1467.85) ;
\draw [shift={(274.3,1468.93)}, rotate = 231.92] [color={rgb, 255:red, 0; green, 0; blue, 0 }  ][line width=0.75]    (10.93,-3.29) .. controls (6.95,-1.4) and (3.31,-0.3) .. (0,0) .. controls (3.31,0.3) and (6.95,1.4) .. (10.93,3.29)   ;

\draw (90,1481.73) node [anchor=north west][inner sep=0.75pt]   [align=left] {$w_2$};
\draw (266,1482.73) node [anchor=north west][inner sep=0.75pt]   [align=left] {$w_3$};
\draw (177.33,1600.06) node [anchor=north west][inner sep=0.75pt]   [align=left] {$w_1$};
\draw (140,1537.39) node   [align=left] {\begin{minipage}[lt]{68pt}\setlength\topsep{0pt}
$\bw$
\end{minipage}};
\draw (282.33,1525.06) node [anchor=north west][inner sep=0.75pt]   [align=left] {$B3^s$};
\draw (310.67,1620.06) node   [align=left] {\begin{minipage}[lt]{68pt}\setlength\topsep{0pt}
$L6$
\end{minipage}};

\end{tikzpicture}

%% file: ex/ex10a.tex
\tikzset{every picture/.style={line width=0.75pt}} 

\begin{tikzpicture}[x=0.75pt,y=0.75pt,yscale=-0.5,xscale=0.5]

\draw   (157,1019.37) .. controls (157,1005.57) and (168.19,994.37) .. (182,994.37) .. controls (195.81,994.37) and (207,1005.57) .. (207,1019.37) .. controls (207,1033.18) and (195.81,1044.37) .. (182,1044.37) .. controls (168.19,1044.37) and (157,1033.18) .. (157,1019.37) -- cycle ;
\draw    (180.75,1046.61) -- (181,1157.37) ;
\draw [shift={(180.75,1044.61)}, rotate = 89.87] [color={rgb, 255:red, 0; green, 0; blue, 0 }  ][line width=0.75]    (10.93,-3.29) .. controls (6.95,-1.4) and (3.31,-0.3) .. (0,0) .. controls (3.31,0.3) and (6.95,1.4) .. (10.93,3.29)   ;
\draw   (157,1182.37) .. controls (157,1168.57) and (168.19,1157.37) .. (182,1157.37) .. controls (195.81,1157.37) and (207,1168.57) .. (207,1182.37) .. controls (207,1196.18) and (195.81,1207.37) .. (182,1207.37) .. controls (168.19,1207.37) and (157,1196.18) .. (157,1182.37) -- cycle ;
\draw    (187.36,994.33) .. controls (214.3,944.86) and (125.71,959.39) .. (170.37,995.23) ;
\draw [shift={(171.76,996.33)}, rotate = 217.44] [color={rgb, 255:red, 0; green, 0; blue, 0 }  ][line width=0.75]    (10.93,-3.29) .. controls (6.95,-1.4) and (3.31,-0.3) .. (0,0) .. controls (3.31,0.3) and (6.95,1.4) .. (10.93,3.29)   ;
\draw    (174.06,1206.67) .. controls (140.42,1247) and (219.14,1254.6) .. (193.67,1205.38) ;
\draw [shift={(192.86,1203.87)}, rotate = 61.28] [color={rgb, 255:red, 0; green, 0; blue, 0 }  ][line width=0.75]    (10.93,-3.29) .. controls (6.95,-1.4) and (3.31,-0.3) .. (0,0) .. controls (3.31,0.3) and (6.95,1.4) .. (10.93,3.29)   ;
\draw (166,1007.22) node [anchor=north west][inner sep=0.75pt]   [align=left] {$w_2$};
\draw (298.33,1007.22) node   [align=left] {\begin{minipage}[lt]{68pt}\setlength\topsep{0pt}
$\ns$
\end{minipage}};
\draw (298.33,1170.3) node   [align=left] {\begin{minipage}[lt]{68pt}\setlength\topsep{0pt}
$L6$
\end{minipage}};
\draw (166,1170.3) node [anchor=north west][inner sep=0.75pt]   [align=left] {$w_1$};

\end{tikzpicture}

%% file: ex/ex10b.tex
 
\tikzset{every picture/.style={line width=0.75pt}} 

\begin{tikzpicture}[x=0.75pt,y=0.75pt,yscale=-0.5,xscale=0.5]

\draw   (157,1019.37) .. controls (157,1005.57) and (168.19,994.37) .. (182,994.37) .. controls (195.81,994.37) and (207,1005.57) .. (207,1019.37) .. controls (207,1033.18) and (195.81,1044.37) .. (182,1044.37) .. controls (168.19,1044.37) and (157,1033.18) .. (157,1019.37) -- cycle ;
\draw    (180.75,1046.61) -- (181,1157.37) ;
\draw [shift={(180.75,1044.61)}, rotate = 89.87] [color={rgb, 255:red, 0; green, 0; blue, 0 }  ][line width=0.75]    (10.93,-3.29) .. controls (6.95,-1.4) and (3.31,-0.3) .. (0,0) .. controls (3.31,0.3) and (6.95,1.4) .. (10.93,3.29)   ;
\draw   (157,1182.37) .. controls (157,1168.57) and (168.19,1157.37) .. (182,1157.37) .. controls (195.81,1157.37) and (207,1168.57) .. (207,1182.37) .. controls (207,1196.18) and (195.81,1207.37) .. (182,1207.37) .. controls (168.19,1207.37) and (157,1196.18) .. (157,1182.37) -- cycle ;
\draw    (187.36,994.33) .. controls (214.3,944.86) and (125.71,959.39) .. (170.37,995.23) ;
\draw [shift={(171.76,996.33)}, rotate = 217.44] [color={rgb, 255:red, 0; green, 0; blue, 0 }  ][line width=0.75]    (10.93,-3.29) .. controls (6.95,-1.4) and (3.31,-0.3) .. (0,0) .. controls (3.31,0.3) and (6.95,1.4) .. (10.93,3.29)   ;
\draw    (174.06,1206.67) .. controls (140.42,1247) and (219.14,1254.6) .. (193.67,1205.38) ;
\draw [shift={(192.86,1203.87)}, rotate = 61.28] [color={rgb, 255:red, 0; green, 0; blue, 0 }  ][line width=0.75]    (10.93,-3.29) .. controls (6.95,-1.4) and (3.31,-0.3) .. (0,0) .. controls (3.31,0.3) and (6.95,1.4) .. (10.93,3.29)   ;

\draw (166,1007.22) node [anchor=north west][inner sep=0.75pt]   [align=left] {$w_4$};
\draw (298.33,1007.22) node   [align=left] {\begin{minipage}[lt]{68pt}\setlength\topsep{0pt}
$\ns$
\end{minipage}};
\draw (298.33,1170.3) node   [align=left] {\begin{minipage}[lt]{68pt}\setlength\topsep{0pt}
$L6$
\end{minipage}};
\draw (166,1170.3) node [anchor=north west][inner sep=0.75pt]   [align=left] {$w_3$};

\end{tikzpicture}

%% file: ex/ex10c.tex
\tikzset{every picture/.style={line width=0.6pt}} 

\begin{tikzpicture}[x=0.75pt,y=0.75pt,yscale=-0.5,xscale=0.5]

\draw   (157,1019.37) .. controls (157,1005.57) and (168.19,994.37) .. (182,994.37) .. controls (195.81,994.37) and (207,1005.57) .. (207,1019.37) .. controls (207,1033.18) and (195.81,1044.37) .. (182,1044.37) .. controls (168.19,1044.37) and (157,1033.18) .. (157,1019.37) -- cycle ;
\draw    (180.75,1046.61) -- (181,1157.37) ;
\draw [shift={(180.75,1044.61)}, rotate = 89.87] [color={rgb, 255:red, 0; green, 0; blue, 0 }  ][line width=0.75]    (10.93,-3.29) .. controls (6.95,-1.4) and (3.31,-0.3) .. (0,0) .. controls (3.31,0.3) and (6.95,1.4) .. (10.93,3.29)   ;
\draw   (157,1182.37) .. controls (157,1168.57) and (168.19,1157.37) .. (182,1157.37) .. controls (195.81,1157.37) and (207,1168.57) .. (207,1182.37) .. controls (207,1196.18) and (195.81,1207.37) .. (182,1207.37) .. controls (168.19,1207.37) and (157,1196.18) .. (157,1182.37) -- cycle ;
\draw    (187.36,994.33) .. controls (214.3,944.86) and (125.71,959.39) .. (170.37,995.23) ;
\draw [shift={(171.76,996.33)}, rotate = 217.44] [color={rgb, 255:red, 0; green, 0; blue, 0 }  ][line width=0.75]    (10.93,-3.29) .. controls (6.95,-1.4) and (3.31,-0.3) .. (0,0) .. controls (3.31,0.3) and (6.95,1.4) .. (10.93,3.29)   ;
\draw    (174.06,1206.67) .. controls (140.42,1247) and (219.14,1254.6) .. (193.67,1205.38) ;
\draw [shift={(192.86,1203.87)}, rotate = 61.28] [color={rgb, 255:red, 0; green, 0; blue, 0 }  ][line width=0.75]    (10.93,-3.29) .. controls (6.95,-1.4) and (3.31,-0.3) .. (0,0) .. controls (3.31,0.3) and (6.95,1.4) .. (10.93,3.29)   ;
\draw (166,1007.22) node [anchor=north west][inner sep=0.75pt]   [align=left] {$w_6$};
\draw (298.33,1007.22) node   [align=left] {\begin{minipage}[lt]{68pt}\setlength\topsep{0pt}
$\ns$
\end{minipage}};
\draw (298.33,1170.3) node   [align=left] {\begin{minipage}[lt]{68pt}\setlength\topsep{0pt}
$L6$
\end{minipage}};
\draw (166,1170.3) node [anchor=north west][inner sep=0.75pt]   [align=left] {$w_5$};

\end{tikzpicture}

%% file: ex/ex10d.tex
\tikzset{every picture/.style={line width=0.75pt}} 

\begin{tikzpicture}[x=0.75pt,y=0.75pt,yscale=-0.5,xscale=0.5]

\draw   (157,1019.37) .. controls (157,1005.57) and (168.19,994.37) .. (182,994.37) .. controls (195.81,994.37) and (207,1005.57) .. (207,1019.37) .. controls (207,1033.18) and (195.81,1044.37) .. (182,1044.37) .. controls (168.19,1044.37) and (157,1033.18) .. (157,1019.37) -- cycle ;
\draw    (180.75,1046.61) -- (181,1157.37) ;
\draw [shift={(180.75,1044.61)}, rotate = 89.87] [color={rgb, 255:red, 0; green, 0; blue, 0 }  ][line width=0.75]    (10.93,-3.29) .. controls (6.95,-1.4) and (3.31,-0.3) .. (0,0) .. controls (3.31,0.3) and (6.95,1.4) .. (10.93,3.29)   ;
\draw   (157,1182.37) .. controls (157,1168.57) and (168.19,1157.37) .. (182,1157.37) .. controls (195.81,1157.37) and (207,1168.57) .. (207,1182.37) .. controls (207,1196.18) and (195.81,1207.37) .. (182,1207.37) .. controls (168.19,1207.37) and (157,1196.18) .. (157,1182.37) -- cycle ;
\draw    (187.36,994.33) .. controls (214.3,944.86) and (125.71,959.39) .. (170.37,995.23) ;
\draw [shift={(171.76,996.33)}, rotate = 217.44] [color={rgb, 255:red, 0; green, 0; blue, 0 }  ][line width=0.75]    (10.93,-3.29) .. controls (6.95,-1.4) and (3.31,-0.3) .. (0,0) .. controls (3.31,0.3) and (6.95,1.4) .. (10.93,3.29)   ;
\draw    (174.06,1206.67) .. controls (140.42,1247) and (219.14,1254.6) .. (193.67,1205.38) ;
\draw [shift={(192.86,1203.87)}, rotate = 61.28] [color={rgb, 255:red, 0; green, 0; blue, 0 }  ][line width=0.75]    (10.93,-3.29) .. controls (6.95,-1.4) and (3.31,-0.3) .. (0,0) .. controls (3.31,0.3) and (6.95,1.4) .. (10.93,3.29)   ;

\draw (166,1007.22) node [anchor=north west][inner sep=0.75pt]   [align=left] {$w_8$};
\draw (298.33,1007.22) node   [align=left] {\begin{minipage}[lt]{68pt}\setlength\topsep{0pt}
$\ns$
\end{minipage}};
\draw (298.33,1170.3) node   [align=left] {\begin{minipage}[lt]{68pt}\setlength\topsep{0pt}
$L6$
\end{minipage}};
\draw (166,1170.3) node [anchor=north west][inner sep=0.75pt]   [align=left] {$w_7$};

\end{tikzpicture}

%% file: ex/ex.T.tex
\tikzset{every picture/.style={line width=0.5pt}} 

\begin{tikzpicture}[x=0.75pt,y=0.75pt,yscale=-0.7,xscale=0.7]

\draw    (23,76) -- (131.11,113.35) ;
\draw [shift={(133,114)}, rotate = 199.06] [color={rgb, 255:red, 0; green, 0; blue, 0 }  ][line width=0.75]    (10.93,-3.29) .. controls (6.95,-1.4) and (3.31,-0.3) .. (0,0) .. controls (3.31,0.3) and (6.95,1.4) .. (10.93,3.29)   ;
\draw    (21,54.5) -- (124.13,15.7) ;
\draw [shift={(126,15)}, rotate = 159.38] [color={rgb, 255:red, 0; green, 0; blue, 0 }  ][line width=0.75]    (10.93,-3.29) .. controls (6.95,-1.4) and (3.31,-0.3) .. (0,0) .. controls (3.31,0.3) and (6.95,1.4) .. (10.93,3.29)   ;
\draw    (154,14.5) -- (263,14.5) ;
\draw [shift={(265,14.5)}, rotate = 180] [color={rgb, 255:red, 0; green, 0; blue, 0 }  ][line width=0.75]    (10.93,-3.29) .. controls (6.95,-1.4) and (3.31,-0.3) .. (0,0) .. controls (3.31,0.3) and (6.95,1.4) .. (10.93,3.29)   ;

\draw (7,56) node [anchor=north west][inner sep=0.75pt]   [align=left] {$t_1$};
\draw (135,6) node [anchor=north west][inner sep=0.75pt]   [align=left] {$t_2$};
\draw (135,105) node [anchor=north west][inner sep=0.75pt]   [align=left] {$t_3$};
\draw (270,6) node [anchor=north west][inner sep=0.75pt]   [align=left] {$t_4$};

\end{tikzpicture}

%% file: ex/ex.last.1.tex
\tikzset{every picture/.style={line width=0.75pt}} 

\begin{tikzpicture}[x=0.75pt,y=0.75pt,yscale=-0.5,xscale=0.5]

\draw   (157,1019.37) .. controls (157,1005.57) and (168.19,994.37) .. (182,994.37) .. controls (195.81,994.37) and (207,1005.57) .. (207,1019.37) .. controls (207,1033.18) and (195.81,1044.37) .. (182,1044.37) .. controls (168.19,1044.37) and (157,1033.18) .. (157,1019.37) -- cycle ;
\draw    (180.75,1046.61) -- (181,1157.37) ;
\draw [shift={(180.75,1044.61)}, rotate = 89.87] [color={rgb, 255:red, 0; green, 0; blue, 0 }  ][line width=0.75]    (10.93,-3.29) .. controls (6.95,-1.4) and (3.31,-0.3) .. (0,0) .. controls (3.31,0.3) and (6.95,1.4) .. (10.93,3.29)   ;
\draw   (157,1182.37) .. controls (157,1168.57) and (168.19,1157.37) .. (182,1157.37) .. controls (195.81,1157.37) and (207,1168.57) .. (207,1182.37) .. controls (207,1196.18) and (195.81,1207.37) .. (182,1207.37) .. controls (168.19,1207.37) and (157,1196.18) .. (157,1182.37) -- cycle ;
\draw    (187.36,994.33) .. controls (214.3,944.86) and (125.71,959.39) .. (170.37,995.23) ;
\draw [shift={(171.76,996.33)}, rotate = 217.44] [color={rgb, 255:red, 0; green, 0; blue, 0 }  ][line width=0.75]    (10.93,-3.29) .. controls (6.95,-1.4) and (3.31,-0.3) .. (0,0) .. controls (3.31,0.3) and (6.95,1.4) .. (10.93,3.29)   ;
\draw    (174.06,1206.67) .. controls (140.42,1247) and (219.14,1254.6) .. (193.67,1205.38) ;
\draw [shift={(192.86,1203.87)}, rotate = 61.28] [color={rgb, 255:red, 0; green, 0; blue, 0 }  ][line width=0.75]    (10.93,-3.29) .. controls (6.95,-1.4) and (3.31,-0.3) .. (0,0) .. controls (3.31,0.3) and (6.95,1.4) .. (10.93,3.29)   ;

\draw (166,1007.22) node [anchor=north west][inner sep=0.75pt]   [align=left] {$w_2$};
\draw (298.33,1007.22) node   [align=left] {\begin{minipage}[lt]{68pt}\setlength\topsep{0pt}
$\bs$
\end{minipage}};
\draw (298.33,1170.3) node   [align=left] {\begin{minipage}[lt]{68pt}\setlength\topsep{0pt}
$L6$
\end{minipage}};
\draw (166,1170.3) node [anchor=north west][inner sep=0.75pt]   [align=left] {$w_1$};

\end{tikzpicture}

%% file: ex/ex.last.2.tex
\tikzset{every picture/.style={line width=0.75pt}} 

\begin{tikzpicture}[x=0.75pt,y=0.75pt,yscale=-0.5,xscale=0.5]

\draw   (157,1019.37) .. controls (157,1005.57) and (168.19,994.37) .. (182,994.37) .. controls (195.81,994.37) and (207,1005.57) .. (207,1019.37) .. controls (207,1033.18) and (195.81,1044.37) .. (182,1044.37) .. controls (168.19,1044.37) and (157,1033.18) .. (157,1019.37) -- cycle ;
\draw    (180.75,1046.61) -- (181,1157.37) ;
\draw [shift={(180.75,1044.61)}, rotate = 89.87] [color={rgb, 255:red, 0; green, 0; blue, 0 }  ][line width=0.75]    (10.93,-3.29) .. controls (6.95,-1.4) and (3.31,-0.3) .. (0,0) .. controls (3.31,0.3) and (6.95,1.4) .. (10.93,3.29)   ;
\draw   (157,1182.37) .. controls (157,1168.57) and (168.19,1157.37) .. (182,1157.37) .. controls (195.81,1157.37) and (207,1168.57) .. (207,1182.37) .. controls (207,1196.18) and (195.81,1207.37) .. (182,1207.37) .. controls (168.19,1207.37) and (157,1196.18) .. (157,1182.37) -- cycle ;
\draw    (187.36,994.33) .. controls (214.3,944.86) and (125.71,959.39) .. (170.37,995.23) ;
\draw [shift={(171.76,996.33)}, rotate = 217.44] [color={rgb, 255:red, 0; green, 0; blue, 0 }  ][line width=0.75]    (10.93,-3.29) .. controls (6.95,-1.4) and (3.31,-0.3) .. (0,0) .. controls (3.31,0.3) and (6.95,1.4) .. (10.93,3.29)   ;
\draw    (174.06,1206.67) .. controls (140.42,1247) and (219.14,1254.6) .. (193.67,1205.38) ;
\draw [shift={(192.86,1203.87)}, rotate = 61.28] [color={rgb, 255:red, 0; green, 0; blue, 0 }  ][line width=0.75]    (10.93,-3.29) .. controls (6.95,-1.4) and (3.31,-0.3) .. (0,0) .. controls (3.31,0.3) and (6.95,1.4) .. (10.93,3.29)   ;

\draw (166,1007.22) node [anchor=north west][inner sep=0.75pt]   [align=left] {$w_4$};
\draw (298.33,1007.22) node   [align=left] {\begin{minipage}[lt]{68pt}\setlength\topsep{0pt}
$\bs$
\end{minipage}};
\draw (298.33,1170.3) node   [align=left] {\begin{minipage}[lt]{68pt}\setlength\topsep{0pt}
$L6$
\end{minipage}};
\draw (166,1170.3) node [anchor=north west][inner sep=0.75pt]   [align=left] {$w_3$};

\end{tikzpicture}

%% file: ex/ex.last.3.tex
\tikzset{every picture/.style={line width=0.75pt}} 

\begin{tikzpicture}[x=0.75pt,y=0.75pt,yscale=-0.5,xscale=0.5]

\draw   (157,1019.37) .. controls (157,1005.57) and (168.19,994.37) .. (182,994.37) .. controls (195.81,994.37) and (207,1005.57) .. (207,1019.37) .. controls (207,1033.18) and (195.81,1044.37) .. (182,1044.37) .. controls (168.19,1044.37) and (157,1033.18) .. (157,1019.37) -- cycle ;
\draw    (180.75,1046.61) -- (181,1157.37) ;
\draw [shift={(180.75,1044.61)}, rotate = 89.87] [color={rgb, 255:red, 0; green, 0; blue, 0 }  ][line width=0.75]    (10.93,-3.29) .. controls (6.95,-1.4) and (3.31,-0.3) .. (0,0) .. controls (3.31,0.3) and (6.95,1.4) .. (10.93,3.29)   ;
\draw   (157,1182.37) .. controls (157,1168.57) and (168.19,1157.37) .. (182,1157.37) .. controls (195.81,1157.37) and (207,1168.57) .. (207,1182.37) .. controls (207,1196.18) and (195.81,1207.37) .. (182,1207.37) .. controls (168.19,1207.37) and (157,1196.18) .. (157,1182.37) -- cycle ;
\draw    (187.36,994.33) .. controls (214.3,944.86) and (125.71,959.39) .. (170.37,995.23) ;
\draw [shift={(171.76,996.33)}, rotate = 217.44] [color={rgb, 255:red, 0; green, 0; blue, 0 }  ][line width=0.75]    (10.93,-3.29) .. controls (6.95,-1.4) and (3.31,-0.3) .. (0,0) .. controls (3.31,0.3) and (6.95,1.4) .. (10.93,3.29)   ;
\draw    (174.06,1206.67) .. controls (140.42,1247) and (219.14,1254.6) .. (193.67,1205.38) ;
\draw [shift={(192.86,1203.87)}, rotate = 61.28] [color={rgb, 255:red, 0; green, 0; blue, 0 }  ][line width=0.75]    (10.93,-3.29) .. controls (6.95,-1.4) and (3.31,-0.3) .. (0,0) .. controls (3.31,0.3) and (6.95,1.4) .. (10.93,3.29)   ;

\draw (166,1007.22) node [anchor=north west][inner sep=0.75pt]   [align=left] {$w_6$};
\draw (298.33,1007.22) node   [align=left] {\begin{minipage}[lt]{68pt}\setlength\topsep{0pt}
$\bs$
\end{minipage}};
\draw (298.33,1170.3) node   [align=left] {\begin{minipage}[lt]{68pt}\setlength\topsep{0pt}
$L6$
\end{minipage}};
\draw (166,1170.3) node [anchor=north west][inner sep=0.75pt]   [align=left] {$w_5$};

\end{tikzpicture}

%% file: ex/ex.tree.last.tex
\tikzset{every picture/.style={line width=0.75pt}} 

\begin{tikzpicture}[x=0.75pt,y=0.75pt,yscale=-1,xscale=1]

\draw    (23,76) -- (131.11,113.35) ;
\draw [shift={(133,114)}, rotate = 199.06] [color={rgb, 255:red, 0; green, 0; blue, 0 }  ][line width=0.75]    (10.93,-3.29) .. controls (6.95,-1.4) and (3.31,-0.3) .. (0,0) .. controls (3.31,0.3) and (6.95,1.4) .. (10.93,3.29)   ;
\draw    (21,54.5) -- (124.13,15.7) ;
\draw [shift={(126,15)}, rotate = 159.38] [color={rgb, 255:red, 0; green, 0; blue, 0 }  ][line width=0.75]    (10.93,-3.29) .. controls (6.95,-1.4) and (3.31,-0.3) .. (0,0) .. controls (3.31,0.3) and (6.95,1.4) .. (10.93,3.29)   ;

\draw (7,56) node [anchor=north west][inner sep=0.75pt]   [align=left] {$t_1$};
\draw (131,6) node [anchor=north west][inner sep=0.75pt]   [align=left] {$t_2$};
\draw (138,105) node [anchor=north west][inner sep=0.75pt]   [align=left] {$t_3$};

\end{tikzpicture}